\documentclass[12pt]{amsart}
\usepackage{amsmath, amsthm, amsfonts, amsbsy, amssymb, upref, enumerate, bigstrut,  color, mathtools, mathrsfs, float, bm}

\usepackage[left=1.3in,top=1.3in,right=1.3in]{geometry}

\usepackage{epsfig, graphicx}
\usepackage{ hyperref}

\newtheorem{theorem}{Theorem}[section]
\newtheorem{lemma}[theorem]{Lemma}

\newtheorem{proposition}[theorem]{Proposition}

\newtheorem{remark}[theorem]{Remark}

\newcommand{\ra}{\rightarrow}

\newcommand{\smallavg}[1]{\langle #1 \rangle}

\newcommand{\fpar}[2]{\frac{\partial #1}{\partial #2}}

\newcommand{\mpar}[3]{\frac{\partial^2 #1}{\partial #2 \partial #3}}

\newcommand{\nto}{\mbox{$\;\rightarrow_{\hspace*{-0.3cm}{\small n}}\;$~}}

\newtheorem{thm1}{Theorem} 

\begin{document}
\title[Absence of Replica Symmetry Breaking in RFIM]
{{Absence of Replica Symmetry Breaking in Finite Fifth Moment Random Field Ising Model}}
\author{J. Roldan}
\address{
\newline
Departamento de Matem\'atica -
Universidade de Bras\'{i}lia, Brazil,
Email: \textup{\tt jinsupe10000@gmail.com}
}
%
\author{R. Vila$^\dag$}
\address{
\newline 
Departamento de Estat\'istica -
Universidade de Bras\'ilia, Brazil,
Email: \textup{\tt rovig161@gmail.com}
}


\keywords{Random Field Ising Model, Replica symmetry.
\\	
$^\dag$ \  Corresponding author.}

\begin{abstract}
This work is concerned with the theory of the Random Field Ising Model on the hypercubic lattice, in the presence of a
independent disorder with finite fifth moment.
We showed the absence of replica symmetry in any dimensions, 
at any temperature and field strength, almost surely. 
\end{abstract}

\maketitle
\section{Introduction}
\label{sec:1}
The Random Field Ising Model (RFIM) \cite{imry1975random,Larkin1970} is probably one of the simplest non-trivial models 
in Statistical Mechanics that belongs to a
class of disordered spin models in which the disorder, so-called quenched random magnetic field, is
coupled to the order parameter of the system. 
This model is under intensive investigation both 
experimentally \cite{PhysRevB.31.4538} and theoretically and until now a lot has
been studied on various aspects, especially the study of 
the existence of phase transition. 
The earliest attempt to address the question of phase transition in the RFIM
goes back to Imry and Ma (1975) \cite{imry1975random}. 
They proposed an extension of the famous Peierls argument to study phase transition in this model.
Following their arguments for $d\leqslant 2$ the uniqueness of the Gibbs states would be expected, while for 
$d\geqslant 3$ this model should have phase transition. 
Remarkable progress on this problem in dimensions $d\geqslant 3$ 
was made by Imbrie \cite{imbrie1984lower,imbrie1985ground} and  
it was  finally solved by Bricmont and Kupiainen (1987-1988) \cite{bricmont1987lower,bricmont1988phase} 
using the renormalization group. 
Subsequently the case $d\leqslant 2$  was solved by Aizenman and Wehr  (1989-1990)
\cite{aizenman1989rounding,aizenman1990rounding}. 	
Recently, there are many studies about the properties of the RFIM, such as decay of correlations, phase transitions
(see, e.g., \cite{aizenman2018power,bricmont1988phase,camia2018note,chatterjee2018decay,Itoi2018})
and 
the replica symmetry breaking (see, e.g., 
\cite{chatterjee2015absence,mezard1992replica,parisi2002physical,talagrand2003spin}).

There are several methods to study disordered systems like the Sherrington-Kirkpatrick model 
\cite{nishimori2001statistical} and the RFIM - the cavity and replica methods  \cite{edwards1975theory} being among them. The replica method enables us to calculate the 
disorder-averaged value of $\ln Z$, where $Z$ denotes the partition function,
in simpler ways, calculating the 
disorder average $Z^n$.
In other words, it uses
$n$ copies or replicas of the system 
(see \cite{spicastellani2005n,nishimori2001statistical}). 
This method works with $n$ as an integer number but later,
by using an analytic continuation to the real numbers, 
analysis was carried out for $n$ tending to $0$.
Using this method the solution is given in terms of its replicas and is known as the replica-symmetric solution. The direct use of the replica method can result in non-physical conclusions. To avoid that, the replica symmetry breaking scheme (see \cite{mezard1987spin,parisi1980magnetic}) is used.
In order to use this scheme we first need to know if the model has this property. Taking this into consideration, it is common for disorder systems like spin glasses to analyze the replica symmetry breaking scheme. 
There are also many contributions from theoretical physics viewpoint to the study of the RFIM, 
see for example \cite{mezard1992replica}.
Chatterjee \cite{chatterjee2015absence}, using the definition of Parisi \cite{parisi2002physical}, showed that  
this model does not have the replica symmetry breaking. 
Later, this same property was proven for 
the transverse and longitudinal RFIM in \cite{Itoi2018}.

Among the main contributions of this paper is the introduction of a class of
independent disorders with finite 5-th moment  where similar results of 
\cite{chatterjee2015absence} can be recovered by
a generalized Gaussian integration by parts 
and Guirlanda-Guerra identities 
(see \cite{aizenmancontucci98,ghirlanda1998general} for a literature review).
In other words,
we show that
the absence of replica symmetry breaking 
is still valid for the model proposed in this paper.
Most of the works in the RFIM  assumes Gaussianity of the disorders 
(see \cite{aizenman2018power,camia2018note,chatterjee2015absence,chatterjee2018decay}). 
In \cite{chatterjee2018decay}, at the end of its introduction, the author warns about 
the difficulty to adapt Gaussian methods and results for models in the presence of disorders with distribution that is not necessarily Gaussian. 
On the other hand, we emphasize that, combining
the results of this paper 
with the main theorem of Panchenko (2011) \cite{Panchenkoultra} we implicitly
established the Parisi ultrametricity property
(see \cite{Panchenkoultra,parisi1980sequence}) in our RFIM. 

	It is known that the presence of disorder in condensed matter systems give rises to 
	new phases and phase transitions possibly related with the multiplicity of metastable states. 
	A system is said to be in a spin glass phase if and only if the ferromagnetic susceptibility 
	is finite, while the spin glass susceptibility is infinite.	
	A long-standing debate on the presence or absence of (elusive) spin glass phase in several systems,
	such as
	the RFIM and the Ginzburg-Landau model (or the so-called $\phi^4$-theory),
	has been the subject of research for many years.
	In \cite{KFRZ} and  \cite{Krzakala2011},
	the authors stated that the RFIM and the Ginzburg-Landau model, respectively, with non-negative interactions and arbitrary disorder on an arbitrary lattice does not have a 
	spin glass phase. That is, they argued that on both ferromagnetic systems the spin glass
	susceptibility is always upper-bounded by the ferromagnetic susceptibility, consequently
	excluding the possibility of a spin glass phase.
	In 2015, Chatterjee \cite{chatterjee2015absence} gave a rigorous mathematical proof 
	supporting part of the findings claimed in \cite{KFRZ}, in the random (Gaussian) field Ising model on the hypercubic lattice.
	Here, we extended the Chatterjee's results to a large class of not necessarily Gaussian random fields, thus giving a further partial support in favor of the absence of a spin glass phase.


This paper is organized as follows. In Section \ref{sec:3}, we begin by presenting the RFIM  and setting up some basic definitions. 
Furthermore, in this section an extension of the main result of 
\cite{chatterjee2015absence} in a more general setting is stated.
In Section \ref{Outline of the proof}, an outline of the proof of this extension is given and
in Section \ref{proof} the proof itself is presented, in details
(see Theorem \ref{rsbthm}). 
%
We end this paper with the proof of our main tool 
in Appendix (see Proposition \ref{prop-cv}).

%
\section{The model}
\label{sec:3}	
Given $n\geqslant 1$, let 
$V_n=\mathbb{Z}^d\cap[1,n]^d$, $d\geqslant 1$, be a finite subset of vertices of 
$d$-dimensional hypercubic lattice 
with cardinality denoted by $|V_n|$.
The (random) Gibbs measure of the RFIM on the set of spin configurations  
$\{-1,1\}^{V_n}$  is given by
\begin{align}\label{gibbs-measure2}
{G}_n(\{\sigma\})
=
{1\over Z_{n}}
\exp\big(
-H_n(\sigma)
\big)\,,
\end{align}
where $H_n$ is a Hamiltonian on $\{-1,1\}^{V_n}$ given by
\begin{align}\label{hamiltonian}
-H_n(\sigma)\coloneqq
\beta \sum_{\langle xy\rangle}\sigma_x \sigma_y + h\sum_{x} g_x \sigma_x\,.
\end{align}
Here, $\langle xy\rangle$ below the first sum means that we are summing over $x, y$ that are neighbors,
$\beta$ and $h$ are positive parameters, called inverse temperature and field strength, respectively.
The partition function $Z_{n}=Z_{n}(\beta,h)$ enters the definition of 
${G}_n$ as a normalizing factor and the
$g_x$'s are
independent 
random variables 
(that collectively are called the disorder)
of form
\begin{align}\label{disorder}
g_x
\coloneqq
h_x \zeta_x, \ \forall x\in\mathbb{Z}^d, \qquad
\sup_{x\in\mathbb{Z}^d}|h_x|\leqslant 1\,, 
\end{align}
where 
$(\zeta_x)$ is an arbitrary disorder with the following properties: the $\zeta_x$'s are  
independent identically distributed (i.i.d.) real-valued random variables with 
zero-mean and unit-variance such that
$\zeta_{x}^5$ is integrable.
Furthermore, the sequence
$(h_x)$ is  a non-zero 
non-invariant 
magnetic external field
such that
\begin{align}\label{condition-main}
\sum_{x\in V_n} |h_x|
=
o(|V_n|),
\quad 
\text{as} \ n\to\infty\,.
\end{align}
%
One can take, for example,
$h_x=h^*\,\|x\|^{-\alpha}$, for $x\neq 0$ and $h_0=h^*$,
with 
$\alpha>0$,
$h^*\in(0,1)$ fixed,
where $\|x-y\|$ denotes the distance between $x$ and $y$ on 
the hypercubic lattice. 
Indeed, 
if $\alpha> d$, $(h_x)$ is summable then \eqref{condition-main} follows trivially;
also when $\alpha\leqslant d$, the external field is not summable and in this case  the condition 
\eqref{condition-main} is obtained by counting over the sizes of $\|x\|$, i.e. doing $\|x\|=r$, with  $r=1,\,2, \,3,\, \cdots$, and next by using the Stolz-Ces\`aro Theorem.
%
%
The study of  the classical nearest neighbor ferromagnetic Ising model in the presence of  
positive power-law decay
external fields with power $\alpha$
appeared recently in several works, see, e.g., 
\cite{Bissacot2015,BISSACOT20174126,Cioletti2016}.
	\begin{remark}
		In the reference \cite{Auffinger}, Example 3, the authors considered the RFIM where 
		the field strength $h$ is a small perturbation (some sort of mean field model). 
		In our setting the field strength remains unchanged with respect to the volume. 
	\end{remark}
\begin{remark}
	Taking $h_x=\pm 1$ and $\zeta_x\sim N(0,1)$ for all $x$, in \eqref{disorder}, 
	the condition \eqref{condition-main} is not satisfied, but this one is not required  
	by the Gaussian integration by parts.
	Then, in this paper we  recovered the results of \cite{chatterjee2015absence} for the 
	RFIM 
	in a general environment.
\end{remark}
\begin{remark}
	The ideal would be to consider the condition: there exist constants $c_1,d_1>0$ such that 
	\begin{align}\label{new-cond}
	c_1 
	\leqslant
	\liminf_{n\to\infty} {1\over |V_n|} \sum_{x\in V_n} |h_x|
	\quad 
	\text{and}
	\quad 
	\limsup_{n\to\infty} {1\over |V_n|} \sum_{x\in V_n} |h_x|
	\leqslant
	d_1\,,
	\end{align}
	instead the condition \eqref{condition-main}. 
	Since the techniques used in this paper include 
	a generalized Gaussian integration by parts,
	we emphasize that unfortunately 
	we can not weaken \eqref{condition-main} by condition \eqref{new-cond}.
	Note also that the magnetics fields $(h_x)$ considered here satisfy the inequality
	of the right side of \eqref{new-cond} with $d_1=1$, but not the condition of the left side.
\end{remark}

\subsection{Some definitions}
For a function $f:(\{-1,1\}^{V_n})^m\to\mathbb{R}$, $m\geqslant 1$, we define
\begin{align}\label{en-int}
\langle \,f\,\rangle
&\coloneqq
\int f(\sigma^1,\ldots,\sigma^m)  \, \text{d}{G}_n(\sigma^1)\cdots \text{d}{G}_n(\sigma^m)
\nonumber
\\[0,2cm]
&= 
\dfrac{1}{Z_n^m}
\sum_{\sigma^1,\ldots,\sigma^m}
f(\sigma^1,\ldots,\sigma^m) 
\exp\sum_{s=1}^{m}\biggl(
\beta \sum_{\langle xy\rangle}\sigma_x^s \sigma_y^s + h\sum_{x} g_x \sigma_x^s
\biggr)
\,.
\end{align}
Let $\langle\,\cdot\,\rangle_{g=u}$ be the Gibbs expectation defined by setting 
$g_x$ in $\langle\,\cdot\,\rangle$ to be $u_x$, for each $x\in V_n$.
The randomness of the $g_x$'s
will be represented by the non-Gaussian measure $\gamma$ on $\mathbb{R}^{n^d}$.
Following the notation of Talagrand \cite{talagrand2003spin}, we write 
\[
\nu(f)
\coloneqq 
\mathbb{E}\langle \,f\,\rangle
=
\int \langle \,f\,\rangle_{g=u} \, \text{d}\gamma(u)
\,,
\]
averaging 
over 
disorder realizations.

If $\sigma^1,\sigma^2,\ldots$ are i.i.d.
configurations under Gibbs measure \eqref{gibbs-measure2}, known as replicas,
the \textit{generalized overlap} (with reference to \cite{chatterjee2009ghirlanda}) 
between two replicas $\sigma^l$, $\sigma^s$ is defined as
\begin{align}\label{overlap}
R_{l,s}= R_{l,s}(\sigma^l, \sigma^s) 
\coloneqq 
{1\over |V_n|}\sum_{x\in V_n} (\mathbb{E}g_x^2)^{1-\delta_{l,s}}\, \sigma^l_x \sigma^s_x, \quad \forall l,s\,,
%
\end{align}
where $\delta$ is the Kronecker delta function and $\mathbb{E}(g_x g_y)=\delta_{x,y}\,h_x h_y$.
Without loss of generality we are assuming that the deterministic constant
$R_{l,l}$ 
is equal to $1$
for all $\sigma^l$.
Note that $|R_{l,s}|\leqslant 1$ (by Schwartz' inequality)
and that the
infinite random array $R = (R_{l,s})_{l,s\geqslant 1}$ is symmetric, 
non-negative definite, weakly
exchangeable (that is, $(R_{l,s})_{1\leqslant l,s\leqslant m}$ and
$(R_{\rho(l),\rho(s)})_{1\leqslant l,s\leqslant m}$ have the same distribution,
for any permutation $\rho:\{1,\ldots,m\}\to \{1,\ldots,m\}$ and for any $m\geqslant 1$).
The array $R$ is said
to satisfy the Ghirlanda-Guerra identities 
(see \cite{aizenmancontucci98,ghirlanda1998general}) if for any $m\geqslant2$ 
and any bounded measurable function $f=f\big((R_{l,s})_{1\leqslant l,s\leqslant m}\big)$,
\begin{align}\label{G-G}
\nu(f R_{1,m+1})
- 
\frac{1}{m}  \,
\nu(f)
\nu(R_{1,2})
- 
\frac{1}{m}
\sum_{s=2}^m \nu(f R_{1,s})
\nto 0\,, 
\end{align}
at almost all $(\beta,h)$.

For each $(\beta,h)\in (0,\infty)^2$, let 
\begin{align}\label{log}
F_{n} = F_{n}(\beta,h) \coloneqq \log Z_{n},
\ 
\psi_{n}= \psi_{n}(\beta, h) \coloneqq \dfrac{F_{n}}{|V_n|},
\ 
p_{n}=p_{n}(\beta, h)\coloneqq\mathbb{E}\psi_{n}\,,
\end{align}
where $\psi_{n}$ is proportional to the free energy
and $p_{n}$ is the value expected of $\psi_{n}$.
It is well-known that in the thermodynamic limit,
$
p = p(\beta,h) \coloneqq \lim_{n\ra\infty} p_{n}
$
is well defined (see Lemma 2.1 in \cite{chatterjee2015absence}),
$p$ is a convex function of $h$ for every fixed $\beta$ and the same is true for $F_{n}$, $\psi_{n}$ and $p_n$ (see Lemma 2.2 in \cite{chatterjee2015absence}).
Then, it is natural to introduce the set
\begin{align}\label{remark-main-1}
{\mathcal A} 
\coloneqq\big\{(\beta,h)\in(0,\infty)^2: 
\textstyle{\partial p\over \partial h^-}(\beta,h)\neq {\partial p\over \partial h^+}(\beta,h) 
\big\}\,.
\end{align}
It is known that, the set ${\mathcal A}$ has zero Lebesgue measure 
and moreover this set is countable (see Lemma 2.3 in \cite{chatterjee2015absence}).

\subsection{The main result}
The system is said to exhibit replica symmetry breaking (as stated in 
\cite{parisi2002physical})
if the limiting
distribution of the random variable $R_{1,2}$, when $n\to\infty$, denoted by $p(q)$ for each $q$,
has more than one point in its support.
In the present paper, 
the next main theorem shows that
this does not happen for the RFIM. 
That is, 
we will show that for any $(\beta, h)\not\in {\mathcal A}$
there exists a constant $q_{\beta,h}\in [-1,1]$ such that
$p(q)$  is a point distribution concentrated on $q_{\beta,h}$.
A sufficient condition to prove this is to verify convergence in quadratic mean,
\begin{align}\label{condition-suf}
\nu\big((R_{1,2}-q_{\beta, h})^2\big)\nto 0\,.
\end{align}
\begin{thm1}[Lack of Replica Symmetry Breaking in the RFIM]\label{rsbthm}
	For any $(\beta, h)\not\in {\mathcal A}$  
	the infinite volume limit: $\lim_{n\to\infty}\nu(R_{1,2})$, exists, and the
	variance of the overlap $R_{1,2}$
	vanishes
	\[
	\nu\big( (R_{1,2}-\nu(R_{1,2}))^2 \big)\nto 0\,.
	\]
	That is, the overlap $R_{1,2}$ is self-averaging in the RFIM  defined by 
	\eqref{gibbs-measure2}.
	Furthermore, there exists a constant $q_{\beta,h}\in [-1,1]$ such that \eqref{condition-suf} is satisfied.
\end{thm1}

The rest of this paper is devoted to the proof of Theorem \ref{rsbthm}.

\section{Outline of the proof}\label{Outline of the proof}
The key results that structure the content of the main theorem of this paper are the following:
\begin{enumerate}
	\item[(I)] 
	{\it The FKG property of the RFIM:} 
	if $f$ and $g$ are two monotone increasing functions on the configuration space $\{-1,1\}^{V_n}$,
	then 
	$
	\smallavg{f g} 
	\geqslant 
	\smallavg{f}\smallavg{g}
	$.
	The proof follows by verifying the FKG lattice condition \cite{fortuin1971correlation} 
	for any realization of the disorder.
	\item[(II)]
	{\it A generalized  Gaussian integration by parts}:
	if $f=f\big((R_{l,s})_{1\leqslant l,s\leqslant m}\big):\mathbb{R}^{m(m-1)/2}$ $\ra[-1,1]$ 
	is a bounded measurable function
	of the overlaps that not change with $n$, using Taylor's Theorem 
	(see Propositions \ref{prop-cheng} and \ref{prop-cv} in Appendix) and the condition
	\eqref{condition-main}, we show that
	\begin{align}
	& \textstyle
	\sum_x
	\mathbb{E}g_x \langle\,\sigma_x^1 \, f\,\rangle 
	-
	\sum_x
	\mathbb{E}g_x^2\, \mathbb{E}{\partial \langle\,\sigma_x^1 \, f\,\rangle\over 
		\partial g_x}
	=
	o(|V_n|)\,, \label{gul}
	\\[0,1cm]
	& \textstyle
	\sum_{x,y}
	\mathbb{E}g_x g_y (F_{n} - \mathbb{E}F_{n})
	-
	\sum_{x,y}
	\mathbb{E}g_x^2 \mathbb{E}g_y^2 \,
	\mathbb{E}\mpar{F_n}{g_x}{g_y}
	=
	o(|V_n|^2)\,, \label{eq-2}
	\\[0,1cm]
	& \textstyle
	\sum_{x,y}
	\mathbb{E}g_x g_y 
	\smallavg{\sigma_x;\sigma_y}
	-
	\sum_{x,y}
	\mathbb{E}g_x^2 \mathbb{E}g_y^2 \,
	\mathbb{E}\mpar{\smallavg{\sigma_x;\sigma_y}}{g_x}{g_y}
	=
	o(|V_n|^2)\,, \label{tcf}
	\end{align}
	in the limit $n\to\infty$, where $F_n$ is as in \eqref{log} and
	$\smallavg{\sigma_x;\sigma_y}
	\coloneqq \smallavg{\sigma_x\sigma_y}
	-
	\smallavg{\sigma_x}\smallavg{\sigma_y}$
	denotes the truncated two-point correlation for the Gibbs measure in finite volume defined by \eqref{gibbs-measure2}.
	\item[(III)]	
	$\mathbb{E}\left(\smallavg{R_{1,2}^2}
	- 
	\smallavg{R_{1,2}}^2 
	\right)=o(1)$  
	and 
	$\nu\big((m(\sigma)-\smallavg{m(\sigma)})^2\big)=o(1)$, 
	as $n\to\infty$ (see Lemma \ref{mainlmm-il}), 
	where $R_{1,2}$ is as in \eqref{overlap} and $m(\sigma)\coloneqq \sum_x\sigma_x/|V_n|$ 
	defines the magnetization: the proof o this item
	follows by using FKG inequality (Item (I)) and the generalized  Gaussian integration by parts 
	\eqref{eq-2}.
	\item[(IV)]
	$\nu\big(|\Delta_n - \nu(\Delta_n) | \big)=o(1)$,
	as $n\to\infty$, $\forall (\beta,h)\in {\mathcal A}^c$ (see Lemma \ref{lemma-main}),
	where $\Delta_n= \sum_{x\in V_n} g_x \sigma_x/|V_n|$:
	following the same steps as Lemma 2.7 in \cite{chatterjee2015absence}, we obtain that
	$\forall(\beta, h)\in {\mathcal A}^c$, 
	$
	\nu(\Delta_n)
	\nto
	\fpar{p}{h}(\beta,h)
	$
	and
	$
	\mathbb{E}|\smallavg{\Delta_n}-\nu(\Delta_n)| 
	\nto 0.
	$
	Combining these two limits with \eqref{tcf} the proof of this item follows.
	\item[(V)] {\it The Ghirlanda-Guerra identities} (see Lemma \ref{ggid}):
	for any bounded measurable function 
	$f:\mathbb{R}^{m(m-1)/2}\ra[-1,1]$, using \eqref{gul}, we show that
	\begin{align}\label{1-appx}
	\nu(\Delta_n(\sigma^1) f)
	-
	h\, \nu\biggl(\Big(\sum_{s=1}^{m}R_{1,s}-mR_{1,m+1}\Big)f\biggr)
	=
	O(|V_n|)\,.
	\end{align}
	Equivalently, the difference
	$
	\nu\big(\Delta_n(\sigma^1) f\big)
	- \nu(\Delta_n) \nu(f)
	$
	approximates
	\[
	h\,\nu\biggl(\Big(\nu(R_{1,2})+\sum_{s=2}^{m}R_{1,s}-mR_{1,m+1}\Big)f \biggr)\,,
	\]
	as $n\to\infty$.
	We emphasize that the last expression is exactly equal to zero under the Gaussian disorder by the 
	Gaussian integration by parts.
	Combining the last approximation with the limit 
	$\mathbb{E}|\smallavg{\Delta_n}-\nu(\Delta_n)| \nto 0$ given in (IV), we conclude that
	the Ghirlanda-Guerra identities \eqref{G-G} are satisfied.
\end{enumerate}
Finally, 
using properties of symmetry of overlaps $R_{l,s}$, Items (V) and (III), it follows that
the variance of 
$R_{1,2}$, with respect to $\nu$, converges to zero as $n\to\infty$. Then it is enough to proof that
the limit of $\nu(R_{1,2})$, as $n\to\infty$, exists. Indeed, taking $f=m=1$ in \eqref{1-appx} and
using the limit 
$
\nu(\Delta_n)
\nto
\fpar{p}{h}(\beta,h)
$
given in (IV),
we obtain that $\nu(\Delta_n)=h\big(1-\nu(R_{1,2})\big)+ o(|V_n|)\nto \fpar{p}{h}(\beta,h)$ and the proof
of theorem
follows by taking $q_{\beta,h}\coloneqq 1-{1\over h}\fpar{p}{h}(\beta,h)$.

\break 
\section{Proof}\label{proof}
Before presenting the proof of the main theorem of this paper, 
we state and prove some preliminaries results (propositions and lemmas).

	The proof of the next result was inspired in the arguments of the proof of  Lemma 2.5 in
	the preprint \cite{Itoi-2019}.
	\begin{lemma}\label{Lemma-1}
		There are a function $\theta:\mathbb{N}\to\mathbb{R}$ and a 
		constant $C>0$
		such that
		\begin{align*}
		\mathrm{Var}(F_n)
		\leqslant 
		(Ch|V_n|+\theta(n))h 
		\,.
		\end{align*}
	\end{lemma}
	\begin{proof}
		For all $s\in[0,1]$ we consider a new random field $G=(G_{x})$ given by
		\[
		G_x=G_{x,s}\coloneqq \sqrt{s}\,g_x+\sqrt{1-s}\,g_x^{*}, \quad \forall x\in V_n\,,
		\]
		where $g=(g_x)$ and $g^{*}=(g_x^{*})$ consist of independent random variables.
		We define also a generating function
		\[
		\gamma_n(s)\coloneqq \mathbb{E}[\mathbb{E}^{*}F_n(G)]^2,
		\quad \text{with} \
		F_n(G)=\log Z_n(G)\,,
		\]
		where $\mathbb{E}$ e $\mathbb{E}^{*}$ denote expectation over $g$ and $g^*$, respectively.
		
		A straightforward computation shows that
		\begin{align}\label{id-1}
		\textstyle
		{{\rm d}\gamma_n\over{\rm d} s}(s)
		&=
		\textstyle
		\sum_x
		\mathbb{E}
		\left[
		{g_x\over\sqrt{s}}\,
		\mathbb{E}^{*} F_n(G) 
		\,
		\mathbb{E}^{*}
		{\partial F_n(G)\over \partial G_x}
		-
		\mathbb{E}^{*} F_n(G) 
		\,
		\mathbb{E}^{*} {g_x^{*}\over\sqrt{1-s}}
		{\partial F_n(G)\over \partial G_x}
		\right]\,.
		\end{align}
		Let  $f(G)\big|_{g_x=u}$ be the function defined by setting $g_x$ in 
		$f(G)$ to be $u$ for all Borel mensurable function $f$ depending of the 
		disorder $G$, and 
		\[ \textstyle
		F_{x}(u)\coloneqq 
		\mathbb{E}^{*} F_n(G) 
		\,
		\mathbb{E}^{*}
		{\partial F_n(G)\over \partial G_x}
		\Big|_{g_x=u} 
		\quad \text{and} \quad
		F^*_{x}(v)\coloneqq {\partial F_n(G)\over \partial G_x}
		\Big|_{g_x^*=v}\,,
		\]
		for each $u,v\in\mathbb{R}$. 
		Furthermore, let $\langle\,\cdot\,\rangle_{\!_G}$ be the Gibbs expectation $\langle\,\cdot\,\rangle$ 
		defined by setting the disorder $G$ instead of one $g$, and $f_x(u)\coloneqq \mathbb{E}F_{x}(u)$ and 
		$f_x^*(v)\coloneqq \mathbb{E}F^*_{x}(v)$.
		Since
		\begin{align}\label{ineq-0}
		& \textstyle
		{\partial F_n(G)\over \partial G_x}={h \langle\sigma_x\rangle_{\!_G}}\,,
		\quad 
		{\partial^2 F_n(G)\over \partial G_x^2}=h^2 
		\big(\langle\sigma_x^2\rangle_{\!_G}-\langle\sigma_x\rangle_{\!_G}^2\big)\,,
		\\
		& \textstyle
		{\partial^3 F_n(G)\over \partial G_x^3}=
		-2h^3\langle\sigma_x\rangle_{\!_G} 
		\big(\langle\sigma_x^2\rangle_{\!_G}-\langle\sigma_x\rangle_{\!_G}^2\big)\,,
		\nonumber
		\\
		& \textstyle
		{\partial^4 F_n(G)\over \partial G_x^4}=
		4h^4
		\big(\langle\sigma_x\rangle_{\!_G}^2-{1\over 2}\big)
		\big(\langle\sigma_x^2\rangle_{\!_G}-\langle\sigma_x\rangle_{\!_G}^2\big)\,,
		\nonumber
		\end{align}
		the real-valued functions $f_{x}(u)$ and $f_{x}^*(v)$ have bounded continuous third-order derivatives.
		Then, by Proposition \ref{prop-cheng} in Appendix,
		\begin{align*}
		& \textstyle
		\mathbb{E}
		\left[
		{g_x\over\sqrt{s}}\,
		\mathbb{E}^{*} F_n(G) 
		\,
		\mathbb{E}^{*}
		{\partial F_n(G)\over \partial G_x}
		\right]
		=
		h_x^2\,
		\mathbb{E}
		\left[
		{1\over\sqrt{s}}\,
		{\partial\over \partial g_x}
		\mathbb{E}^{*} F_n(G) 
		\,
		\mathbb{E}^{*}
		{\partial F_n(G)\over \partial G_x}
		\right]
		+
		{1\over \sqrt{s}}
		\boldsymbol{\gamma}^2_{g_x}(f_{x})\,;
		\\[0,1cm]
		& \textstyle
		\mathbb{E}^{*} {g_x^{*}\over\sqrt{1-s}}
		{\partial F_n(G)\over \partial G_x}
		=
		h_x^2\,
		\mathbb{E}^{*}
		{\partial \over \partial g_x^{*}}
		{\partial F_n(G)\over \partial G_x}
		+
		{1\over \sqrt{1-s}}\,
		\boldsymbol{\gamma}^2_{g_x^*}(f_{x}^*)\,,
		\end{align*}
		where 
		$
		\boldsymbol{\gamma}^2_{g_{x}}(f_{x})
		\coloneqq
		\mathbb{E}(
		g_{x}\int_0^{g_x}(g_x-u) {{\rm d}^2f_x\over{\rm d}u^2 }(u) \,{\rm d}u)
		-
		h_x^2\,
		\mathbb{E}(
		\int_0^{g_x}(g_x-u) {{\rm d}^3f_x\over{\rm d}u^3 }(u) \,{\rm d}u).
		$
		Substituting the last two identities in \eqref{id-1}, we have
		\begin{align}\label{ineq-1}
		\textstyle
		{{\rm d}\gamma_n\over{\rm d} s}(s)  
		&=
		\textstyle
		\sum_x
		h_x^2 \,
		\mathbb{E}
		\left[
		{1\over\sqrt{s}}\,
		{\partial\over \partial g_x}
		\mathbb{E}^{*} F_n(G) 
		\,
		\mathbb{E}^{*}
		{\partial F_n(G)\over \partial G_x}
		-
		{1\over\sqrt{1-s}}\,
		\mathbb{E}^{*}F_n(G) 
		\,
		\mathbb{E}^{*}
		{\partial \over \partial g_x^{*}}
		{\partial F_n(G)\over \partial G_x}
		\right]  \nonumber
		\\[0,1cm]
		&\quad +
		{h} \theta(n,s)\,,
		\end{align}
		where $\theta(n,s)\coloneqq	\sum_x [{1\over\sqrt{s}} \boldsymbol{\gamma}^2_{g_x}(f_{x})- {1\over\sqrt{1-s}} \boldsymbol{\gamma}^2_{g_x^{*}}(f_{x}^{*}) ]$.
		
		On other hand, using that
		$
		{\partial F_n(G)\over \partial g_x} 
		=
		\sqrt{s}\, {\partial F_n(G)\over \partial G_x} 
		$
		and
		$
		{\partial F_n(G)\over \partial g_x^{*}} 
		=
		\sqrt{1-s}\, {\partial F_n(G)\over \partial G_x},
		$
		it follows that
		\begin{multline*}
		\textstyle
		{1\over\sqrt{s}}\,
		{\partial \over \partial g_x}
		\mathbb{E}^{*}F_n(G) 
		\,
		\mathbb{E}^{*}
		{\partial F_n(G)\over \partial G_x}
		-
		\mathbb{E}^{*}F_n(G) 
		\,
		\mathbb{E}^{*}{1\over\sqrt{1-s}}
		{\partial \over \partial g_x^{*}}
		{\partial F_n(G)\over \partial G_x}
		=
		\Big(\mathbb{E}^{*}{\partial F_n(G)\over \partial G_x}\Big)^2\,.
		\end{multline*}
		Then, for each $s\in(0,1)$, \eqref{ineq-1} can be rewritten as
		\begin{align}\label{des-main}
		{{\rm d}\gamma_n\over{\rm d} s}(s)
		=
		\sum_x
		h_x^2\,
		\mathbb{E}
		\Big(
		\mathbb{E}^{*}{\partial F_n(G)\over \partial G_x}
		\Big)^2
		+
		{h}
		\theta(n,s)\,.
		\end{align}
		Jensen's inequality, the Item \eqref{ineq-0} and the inequality $h_x^2\leqslant 1$
		for all $x$ give
		\begin{align*}
		{{\rm d}\gamma_n\over{\rm d} s}(s)
		\leqslant
		C_2{h^2}|V_n|
		+
		{h}
		\theta(n,s)
		\,.
		\end{align*}
		Integrating the above inequality from $0$ to $1$ and using the
		relation 
		$\mathrm{Var}(F_n)
		=\gamma_n(1)-\gamma_n(0)=\int_{0}^{1}{{\rm d}\gamma_n\over{\rm d} s}(s)\, \text{d}s$,
		\begin{align*}
		\mathrm{Var}(F_n)
		\leqslant
		C_2h^2|V_n|
		+
		{h}
		\int_{0}^{1}
		\theta(n,s) \, \text{d}s\,.
		\end{align*}
		Finally, taking
		$\theta(n)\coloneqq 
		\int_{0}^{1}
		\theta(n,s) \, \text{d}s$ the proof follows.
	\end{proof}

The proof of the next proposition makes essential use of Lemma \ref{Lemma-1}. 
\begin{proposition}\label{mainlmm-1}
	For any $n$ and any $(\beta,h)\in (0,\infty)^2$, there exists an application $\ell:V_n\times V_n\to\mathbb{R}$ such that  
	\[
	\sum_{x,y\in V_n}
	\Big(
	h_x h_y\,
	\mathbb{E}\mpar{F_n}{g_x}{g_y}
	+
	\ell(x,y)
	\Big)^2
	-
	\sum_{x\in V_n} 
	\big(
	h_x^2
	+
	\ell(x,x)
	\big)^2	
	\leqslant 
	2(Ch|V_n|+\theta(n))h 
	\, ,
	\]
	with $C$ and $\theta(\cdot)$ the same constant and function as in Lemma \ref{Lemma-1}, respectively.
\end{proposition}
\begin{proof}
	Let 
	$F \coloneqq F_{n} - \mathbb{E}F_{n}$.
	Let $W$ denote the set of all unordered pairs $\{x,y\}$, where $x,y\in V_n$ and $x\ne y$. 
	For each $e = \{x,y\}\in W$, let 
	$s_e \coloneqq \zeta_x \zeta_y,$
	$c_e \coloneqq \mathbb{E}s_e F$
	\text{and}  
	$S \coloneqq \sum_{e\in W} c_e s_e,$
	where $\zeta_x$ is as in \eqref{disorder}.
	Since $\zeta_x$ has zero-mean and unit-variance for all $x$,
	for any $e\neq e'$ in $W$, $\mathbb{E}s_e s_{e'} =0$, 
	and for any $e$, $\mathbb{E}s_e^2=1$. Thus,
	$
	\mathbb{E}S^2 
	= 
	\sum_{e}c_e^2
	= \sum_e \mathbb{E}c_e s_e F = \mathbb{E}S F.
	$
	Consequently,
	$
	\mathbb{E}F^2 = \mathbb{E}(F-S)^2 + \mathbb{E}S^2\, .
	$
	Then, by Lemma \ref{Lemma-1},
	\begin{equation}\label{cw}
	\sum_{e\in W} c_e^2 
	= 
	\mathbb{E}S^2 
	\leqslant
	\mathbb{E}F^2
	=
	{\rm Var}(F_{n})
	\leqslant 
	(Ch|V_n|+\theta(n))h 
	\, .
	\end{equation}
	Let $F_{x,y}(u,v)$ be the function defined by setting 
	$g_x$ and $g_y$ in $F$ to be $u$ and $v$ respectively, for any $e=\{x,y\}$.
	Applying a generalized Gaussian integration by parts 
	(see Proposition \ref{prop-cv} in Appendix),
	with $f_{x,y}(u,v)\coloneqq \mathbb{E}F_{x,y}(u,v)$, we have
	\begin{align}\label{def-ce}
	h_x h_y c_e 
	=
	\mathbb{E}g_x g_y F
	= 
	h_x^2 h_y^2 \,
	\mathbb{E}\mpar{f_{x,y}}{g_x}{g_y}
	+
	{\boldsymbol{\gamma}^2_{g_x,g_y}(f_{x,y})} \,,
	\end{align}
	where 
	\begin{align}\label{gamma-0}
	&\boldsymbol{\gamma}^2_{g_x,g_y}(f_{x,y}) \noindent
	\\[0,1cm]
	&\coloneqq
	\textstyle
	\mathbb{E}\Big(
	g_x g_y\int_0^{g_x} \int_0^{g_y} (g_x-u)\, 
	\dfrac{\partial^3 f_{x,y}}{\partial u^2\partial v}
	\,{\rm d}u {\rm d}v \Big) \nonumber
	\\[0,1cm]
	&\quad - 
	\textstyle
	h_x^2 h_y^2  \,
	\mathbb{E}\Big(
	\int_0^{g_x} \int_0^{g_y} (g_x-u)\, 
	\dfrac{\partial^5 f_{x,y}}{\partial u^3\partial v^2} \,
	{\rm d}u {\rm d}v
	\Big)
	\nonumber
	\\[0,1cm]
	&\quad -
	\textstyle
	h_x^2 h_y^2  \, 
	\mathbb{E}\Big(
	\int_{0}^{g_x}\!(g_x-u)\,  \dfrac{\partial^4 f_{x,y}(u,0)}{\partial u^3\partial v}
	\text{d}u
	+ \textstyle
	\int_{0}^{g_y}\!(g_y-v)\,  \dfrac{\partial^4 f_{x,y}(0,v)}{\partial u\partial v^3}\text{d}v
	\Big)
	\,. 
	\nonumber
	\end{align}
	Combining 
	\eqref{cw} and
	\eqref{def-ce}, it follows that
	\begin{align*}
	&\sum_{x,y} 
	\Big( 
	h_x h_y \, \mathbb{E}\mpar{f_{x,y}}{g_x}{g_y}
	+
	(h_x h_y)^{-1} \,
	{\boldsymbol{\gamma}^2_{g_x,g_y}(f_{x,y})}
	\Big)^2 
	\\[0,1cm]
	& \leqslant  
	\sum_{x\neq y}
	\Big(\!
	h_x h_y \, \mathbb{E}\mpar{f_{x,y}}{g_x}{g_y}\!
	+
	(h_x h_y)^{-1} \,
	{\boldsymbol{\gamma}^2_{g_x,g_y}(f_{x,y})}
	\!\Big)^2 \!\!\!
	+ 
	\sum_{x} \!
	\big(
	h_x^2
	+
	h_x^{-2} \,
	{\boldsymbol{\gamma}^2_{g_x,g_x}(f_{x,y})}
	\big)^2
	\\[0,1cm]
	&
	\leqslant 
	2(Ch|V_n|+\theta(n))h
	+
	\sum_{x} 
	\big(
	h_x^2
	+
	h_x^{-2} \,
	{\boldsymbol{\gamma}^2_{g_x,g_x}(f_{x,x})}
	\big)^2 \, .
	\end{align*}
	Thus, taking 
	$
	\ell(x,y)
	\coloneqq
	(h_x h_y)^{-1}\, \boldsymbol {\gamma}^2_{g_x,g_y}(f_{x,y})
	$ 
	the proof of proposition follows.
\end{proof}

The FKG property  for the RFIM and the Proposition \ref{mainlmm-1} have
an important role in the proof of the following result.
\begin{proposition}\label{mainlmm}
	For any $n$ and any $(\beta,h)\in (0,\infty)^2$ there exists 
	a sequence $\alpha_n\coloneqq\alpha_n(h)=o(1)$
	such that  
	\begin{align*}
	\mathbb{E}\left(\smallavg{R_{1,2}^2}
	- 
	\smallavg{R_{1,2}}^2 
	\right)
	&\leqslant 
	{4\over h^2}
	\sqrt
	{
		\Big(\alpha_n+{\theta(n)\over |V_n|^2}\Big)h
		+
		{1\over|V_n|^2}\sum_{x} \ell^2(x,x)
	}
	-
	\frac{2\theta_1(n)}{h^2 |V_n|^2}  
	\, ,
	\end{align*}
	where $\theta(\cdot)$ and $\ell(\cdot,\cdot)$ are as in  Lemma \ref{Lemma-1} and
	 Proposition \ref{mainlmm-1}, respectively, and
	$\theta_1(n)\coloneqq \sum_{x,y\in V_n}h_x h_y \, \ell(x,y)$.
\end{proposition}
\begin{proof}
	It is well-known that
	\begin{align}\label{eq-r}
	{1\over h^2} \,
	\mpar{F}{g_x}{g_y}
	=
	\smallavg{\sigma_x;\sigma_y}\,,
	\end{align}
	where $\smallavg{\sigma_x;\sigma_y}$ is the truncated two-point correlation.
	By FKG inequality,
	$\mathbb{E}\smallavg{\sigma_x;\sigma_y}\geqslant 0$ for each $x,y$.
	Then, 
	\begin{align*}
	\mathbb{E}\left(\smallavg{R_{1,2}^2}
	- 
	\smallavg{R_{1,2} }^2 
	\right)
	&= 
	{1\over |V_n|^2}
	\sum_{x,y} h_x^2 h_y^2\,
	\mathbb{E}
	\smallavg{\sigma_x;\sigma_y}(\langle\sigma_x\sigma_y\rangle+\langle\sigma_x\rangle \langle\sigma_y\rangle)
	\\[0,1cm]
	&\leqslant 
	{2\over h^2|V_n|^2} 
	\sum_{x,y} h_x^2 h_y^2\, \mathbb{E}\mpar{F}{g_x}{g_y}
	\, ,
	\end{align*}
	where in the inequality we use \eqref{eq-r} and the upper bound
	$\langle\sigma_x\sigma_y\rangle+\langle\sigma_x\rangle \langle\sigma_y\rangle\leqslant 2$.
	Seeing that $|h_x|\leqslant 1$ for all $x$, by the Cauchy-Schwarz inequality
	the right-hand side term of the above inequality is at most
	\begin{align*} 
	\frac{2}{h^2 |V_n|^2} 
	\left(
	\sqrt{ 
		\sum_{x,y} h_x^2 h_y^2 
		\sum_{x,y} 
		\Big(
		h_x h_y \,
		\mathbb{E}\mpar{F}{g_x}{g_y} +\ell(x,y)
		\Big)^2
	} 
	-
	\theta_1(n)
	\right)\,,
	\end{align*}
	where 	
	$
	\ell(x,y)
	=
	(h_x h_y)^{-1}\, \boldsymbol {\gamma}^2_{g_x,g_y}(f_{x,y}).
	$
	By Proposition \ref{mainlmm-1} and by Minkowski's Inequality, 
	the expression of left side of above difference is at most
	\begin{align*}
	{4\over h^2}
	\sqrt
	{
		\Big({Ch+h^{-1}\over|V_n|}+{\theta(n)\over |V_n|^2}\Big)h
		+
		{1\over|V_n|^2}\sum_{x} \ell^2(x,x)
	}\,.
	\end{align*}
	Therefore, taking $\alpha_n\coloneqq (Ch+h^{-1})/|V_n|,$ the proof of proposition is complete.
\end{proof}
	\begin{remark}
		As mentioned before, 
		the following result was also proved in \cite{Auffinger}, Example 3, in the case that the field strength $h$ is a small perturbation with a decay ratio which similar but different from ours.
	\end{remark}
The inequality provided by Proposition \ref{mainlmm} will allow us the next key lemma. 
\begin{lemma}\label{mainlmm-il}
	For any $(\beta,h)\in (0,\infty)^2$, 
	\[
	\mathbb{E}\left(\smallavg{R_{1,2}^2}
	- 
	\smallavg{R_{1,2}}^2 
	\right)
	\nto 0\,.
	\]
	That is, in mean, 
	the variance of the overlap $R_{1,2}$, with respect to the Gibbs measure \eqref{gibbs-measure2}, 
	converges to zero in the thermodynamic limit.
\end{lemma}
\begin{proof}
	By Proposition \ref{mainlmm} it is enough to prove that
	$\theta(n)=o(|V_n|^2)$, $\theta_1(n)=o(|V_n|^2)$ and $\sum_{x} \ell^2(x,x)=o(|V_n|^2)$. That is,
	\begin{align}
	&\textstyle 
	{1\over |V_n|^2}
	\sum_{x}
	\big(h_x^{-2}\,\boldsymbol{\gamma}^2_{g_x,g_x}(f_{x,x})\big)^2
	\nto 0 \label{first-conv}
	\quad \text{and} \quad
	{1\over |V_n|^2}
	\sum_{x,y}
	\boldsymbol{\gamma}^2_{g_x,g_y}(f_{x,y})
	\nto 0\,,  
	\\[0,1cm]
	&\textstyle 
	{1\over |V_n|^2}
	\int_{0}^{1}	
	\sum_x \big[{1\over\sqrt{s}} \boldsymbol{\gamma}^2_{g_x}(f_{x})- {1\over\sqrt{1-s}} \boldsymbol{\gamma}^2_{g_x^{*}}(f_{x}^{*}) \big]
	\, {\rm d} s
	\nto 0,
	\label{second-conv}
	\end{align}
	where 
	$\boldsymbol{\gamma}^2_{g_x,g_y}(f_{x,y})$
	is as in \eqref{gamma-0}.
	Let us give the details for the convergences in \eqref{first-conv}. The other case \eqref{second-conv} is analogous.
	
	Since 
	$\big|{\partial^{i+j} F_{x,y}(u,v)\over \partial u^i\partial v^j}\big|\leqslant C_{ij} h^{i+j}$,
	for some constant $C_{ij}>0$ depending on the indices $i,j$, it follows that 
	\begin{align}\label{id-n}
	\textstyle
	\Big| 
	\int_0^{g_x} \int_0^{g_y} (g_x-u)\, {\partial^5 f_{x,y}(u,v)\over\partial u^3 \partial v^2} 
	\,{\rm d}u {\rm d}v
	\Big|
	\leqslant 
	{C_{32}h^5 h_x^2|h_y|\over 2}\,
	\zeta_x^2 |\zeta_y|
	\end{align}
	and by 
	Items \eqref{second-id-1}$-$\eqref{second-id-3}  
	of Proposition \ref{prop-cheng}, that
	\begin{align}
	& \textstyle
	\Big| 
	g_x g_y\int_0^{g_x} \int_0^{g_y} (g_x-u)\, {\partial^3 f_{x,y}(u,v)\over\partial u^2 \partial v} \,{\rm d}u {\rm d}v
	\Big|  \label{bound-1}
	\leqslant
	{C_{21}h^3 |h_x|^3h_y^2\over 2} \, |\zeta_x|^3 \zeta_y^2
	\,,  
	\\[0,1cm]
	& \textstyle
	\Big|
	\int_{0}^{g_x}(g_x-u) \,  \dfrac{\partial^4 f_{x,y}(u,0)}{\partial u^3\partial v}  \,\text{d}u
	\Big|
	\leqslant
	{C_{31}h^4 h_x^2\over 2} \, \zeta_x^2
	\,,  \label{bound-3}
	\\[0,1cm]
	& \textstyle
	\Big|
	\int_{0}^{g_y}(g_y-v)\, \dfrac{\partial^4 f_{x,y}(0,v)}{\partial u\partial v^3} \,\text{d}v
	\Big|
	\leqslant
	{C_{13}h^4 h_y^2\over 2} \, \zeta_y^2
	\,.  \label{bound-4}
	\end{align}
	Using the definition of $\boldsymbol{\gamma}^2_{g_x,g_y}(f_{x,y})$
	in \eqref{gamma-0} and taking $x=y$ in the inequalities \eqref{id-n}$-$\eqref{bound-4}, 
	it follows that
	\begin{align*}
	\limsup_{n\to\infty}
	{1\over |V_n|^2}
	\sum_{x} 
	\big(h_x^{-2}\, \boldsymbol{\gamma}^2_{g_x,g_x}(f_{x,y})\big)^2
	\leqslant
	\lim_{n\to\infty}
	{C\over |V_n|^2}
	\sum_{x}
	h_x^4
	(\mathbb{E}|\zeta_x|^5)^2
	=
	0\,,
	\end{align*}
	where $C\coloneqq (C_{21}+C_{32}h^2+C_{31}h+C_{13})^2 h^6/4$.
	The last equality follows
	from \eqref{condition-main} and of the assumption that
	the $\zeta_x$'s are identically distributed and satisfy $\mathbb{E}|\zeta_x|^5<\infty$.
	Therefore, the limit of the left side of \eqref{first-conv} follows.	
	
	Similarly, using the definition of $\boldsymbol{\gamma}^2_{g_x,g_y}(f_{x,y})$
	in \eqref{gamma-0} and the inequalities \eqref{id-n}$-$\eqref{bound-4}, it is proved that
	\begin{align*}
	&\limsup_{n\to\infty}
	{1\over |V_n|^2}
	\sum_{x,y}
	|\boldsymbol{\gamma}^2_{g_x,g_y}(f_{x,y})|
	\leqslant
	\lim_{n\to\infty}
	{\widetilde{C}\over |V_n|^2}
	\sum_{x,y} h_x^2h_y^2
	(\mathbb{E}|\zeta_x|^3 +\mathbb{E}|\zeta_y|+2)=0\,,
	\end{align*}
	where $\widetilde{C}\coloneqq\max\{C_{21}, C_{32}h^2, C_{31}h, C_{13}h\}h^3/2$.
	Again, the last equality follows by hypothesis \eqref{condition-main}  and of the assumption that
	the $\zeta_x$'s are identically distributed and satisfy $\mathbb{E}|\zeta_x|^3<\infty$.
	Then the limit of the right side of \eqref{first-conv} is valid and 
	the proof of lemma is complete.
\end{proof}
%
%
%

The proof of Proposition \ref{mainlmm} plays 
an important role in the proof of the following result.
\begin{proposition}\label{mainlmm-il-1}
	For any $n$ and any $(\beta,h)\in (0,\infty)^2$, 
	\begin{align*}
	{1\over 4h^2}\!
	\sum_{x,y\in V_n}\!
	h_x^2 h_y^2 \,
	\mathbb{E}
	{\partial^4 F_n\over\partial g_x^2\partial g_y^2} 
	\leqslant 
	2|V_n|
	\sqrt{
		\big(\alpha_n|V_n|^2+\theta(n)\big)h
		+\sum_x\ell^2(x,x)
	}
	-
	\theta_1(n)\,,
	\end{align*}
	with $\alpha_n$ as in Proposition \ref{mainlmm}, $\theta(\cdot)$  as in  Lemma \ref{Lemma-1},
	$\ell(\cdot,\cdot)$ as in Proposition \ref{mainlmm-1} and
	$\theta_1(n)=\sum_{x,y\in V_n}h_x h_y \, \ell(x,y)$.
\end{proposition}
\begin{proof}
	As a sub-product of the proof of Proposition \ref{mainlmm} we have
	\begin{multline}\label{1-ident}
	\sum_{x,y}
	h_x^2 h_y^2 \,
	\mathbb{E}
	{\partial^2 F_n\over\partial g_x\partial g_y} 
	\leqslant 
	2|V_n|
	\sqrt{
		\big(\alpha_n|V_n|^2+\theta(n)\big)h
		+\sum_x\ell^2(x,x)
	}
	-
	\theta_1(n)\,,
	\end{multline}
	with
	$
	\ell(x,y)
	=
	(h_x h_y)^{-1}\, \boldsymbol {\gamma}^2_{g_x,g_y}(f_{x,y}).
	$
	
	On the other hand, we claim that
	\begin{align}\label{2-ident}
	{\partial^4 F_n\over\partial g_x^2\partial g_y^2} 
	\leqslant 
	4h^2 
	{\partial^2 F_n\over\partial g_x\partial g_y}\,.
	\end{align}
	Indeed, a straightforward computation shows that
	$
	{\partial\smallavg{\sigma_x;\sigma_y}\over \partial g_x}
	=
	-2h\,\langle\sigma_x\rangle \smallavg{\sigma_x;\sigma_y}\,.
	$
	Then, using the identity \eqref{eq-r}, we have 
	\[
	{\partial^4 F_n\over\partial g_x^2\partial g_y^2} 
	=
	4h^2 
	\Big(
	\langle\sigma_x\rangle \langle\sigma_y\rangle
	-
	{1\over 2}
	\Big)\,
	{\partial^2 F_n\over\partial g_x\partial g_y}\,.
	\]
	Since $\langle\sigma_x\rangle\leqslant 1$ for all $x$, and 
	the derivative 
	${\partial^2 F_n\over\partial g_x\partial g_y}$ is non-negative, the claim follows.
	
	Finally, combining \eqref{1-ident} and \eqref{2-ident} the proof of proposition follows.
\end{proof}

	Next, define
	\[
	\Delta_n=\Delta_n(\sigma)\coloneqq {1\over |V_n|} \sum_{x\in V_n} g_x \sigma_x\, .
	\]
That is, $\Delta_n$ is the part of the energy due to the disorder.
\begin{remark}
	Note that the absolute value  of covariance between $\Delta_n(\sigma^l)$ and $\Delta_n(\sigma^s)$ 
	doesn't grow faster than the absolute value of overlap $R_{l,s}$ \eqref{overlap} between two replicas 
	$\sigma^l$, $\sigma^s$.
\end{remark}

Recall the set ${\mathcal A}$ defined in \eqref{remark-main-1}. 
Let ${\mathcal A}^c$ denote the complement of ${\mathcal A}$ in the cartesian plane. 
\begin{remark}\label{hnlmm}
	Since $\langle \Delta_n \rangle={\partial \psi_n\over \partial h}$, 
	$\nu(\Delta_n)={\partial p_n\over \partial h}$,
	${\rm Var}(F_{n})\leqslant (Ch|V_n|+\theta(n))h$ $($ see Lemma \ref{Lemma-1}$)$ and 
	$p = \lim_{n\ra\infty} p_{n}$ exists and is finite for all $(\beta,h)$, 
	by convexity arguments of the function $h\mapsto\psi_n(\beta,h)$ it follows that:	
	for any $(\beta, h)\in {\mathcal A}^c$, 
	\begin{align*}
	\nu(\Delta_n)
	\nto
	\fpar{p}{h}(\beta,h),
	\quad 
	\mathbb{E}|\smallavg{\Delta_n}-\nu(\Delta_n)| 
	\nto 0\, .
	\end{align*}
	For more details see
	Lemma 2.7 in \cite{chatterjee2015absence}.
\end{remark}
%
%

\begin{lemma}\label{lemma-main}
	For any $(\beta,h)\in {\mathcal A}^c$,
	\[
	\nu\big(|\Delta_n - \nu(\Delta_n) | \big)  \nto 0\, .
	\]
\end{lemma}
\begin{proof}
	Let $\langle\,\cdot\,\rangle_{g_x=u, g_y=v}$ be the Gibbs expectation defined by setting 
	$g_x$ and $g_y$ in $\langle\,\cdot\,\rangle$ to be $u$ and $v$ respectively, and
	$F^*_{x,y}(u,v)\coloneqq \smallavg{\sigma_x;\sigma_y}_{g_x=u,g_y=v}$.	
	A generalized  Gaussian integration by parts (see Proposition \ref{prop-cv} in Appendix),
	with $f^*_{x,y}(u,v)=\mathbb{E}F^*_{x,y}(u,v)$, gives 
	\begin{align*}
	\mathbb{E}g_xg_yf^*_{x,y}
	=
	h_x^2 h_y^2\,
	\mathbb{E}\mpar{f^*_{x,y}}{g_x}{g_y}
	+
	\boldsymbol{\gamma}^2_{g_x,g_y}(f^*_{x,y})\,,
	\end{align*}
	where $\boldsymbol{\gamma}^2_{g_x,g_y}(f^*_{x,y})$
	is defined analogously as in \eqref{gamma-0} with
	$f^*_{x,y}$ instead $f_{x,y}$.
	Dividing this equality by $|V_n|^2$ and summing over all $x,y\in V_n$, and using \eqref{eq-r},
	we obtain
	\begin{align}\label{non-ide}
	\mathbb{E}\big(\smallavg{\Delta_n^2}-\smallavg{\Delta_n}^2\big)
	=
	{1\over h^2|V_n|^2}
	\sum_{x,y}
	\Big(
	h_x^2 h_y^2\,
	\mathbb{E}{\partial^4 F_n\over\partial g_x^2\partial g_y^2} 
	+
	h^2
	\boldsymbol{\gamma}^2_{g_x,g_y}(f^*_{x,y})
	\Big)				
	\,.
	\end{align}
	By Proposition \ref{mainlmm-il-1}, the expression \eqref{non-ide} is at most
	\begin{align*}
	8\sqrt{\Big(\alpha_n+{\theta(n)\over |V_n|^2}\Big)h+ {1\over |V_n|^2}\sum_x \ell^2(x,x)}
	-
	{4\theta_1(n)\over |V_n|^2}
	+
	{1\over |V_n|^2}\sum_{x,y} \gamma^2_{g_x,g_y}(f^*_{x,y})\,,
	%
	\end{align*}
	where $\ell(x,y)=(h_x h_y)^{-1}\,\boldsymbol {\gamma}^2_{g_x,g_y}(F)$.
	That is,
	\begin{align}\label{ineq-e}
	&\mathbb{E}\big(\smallavg{\Delta_n^2}-\smallavg{\Delta_n}^2\big)
	\\[0,1cm]
	&\leqslant
	8\sqrt{\Big(\alpha_n+{\theta(n)\over |V_n|^2}\Big)h+ {1\over |V_n|^2}\sum_x \ell^2(x,x)}
	-
	{4\theta_1(n)\over |V_n|^2}
	+
	{1\over |V_n|^2}\sum_{x,y} \gamma^2_{g_x,g_y}(f^*_{x,y})
	\,. \nonumber
	\end{align}
	Items \eqref{first-conv} and \eqref{second-conv} show that 
	$\theta(n)=o(|V_n|^2)$, $\theta_1(n)=o(|V_n|^2)$ and $\sum_{x} \ell^2(x,x)$
	$ =o(|V_n|^2)$.
	Analogously to the proof of Items \eqref{first-conv}-\eqref{second-conv}, using 
	Items \eqref{second-id-1}$-$\eqref{second-id-3}  
	of Proposition \ref{prop-cheng} in Appendix, it is verified that
	$
	\sum_{x,y}
	\boldsymbol{\gamma}^2_{g_x,g_y}(f^*_{x,y})
	=
	o(|V_n|^2).
	$
	Therefore, since $\alpha_n=o(1)$ (see Proposition \ref{mainlmm}), $\mathbb{E}\big(\smallavg{\Delta_n^2}-\smallavg{\Delta_n}^2\big)=o(1)$.
	
	Combining \eqref{ineq-e} with the inequality
	\[
	\nu(|\Delta_n-\langle \Delta_n\rangle|)
	\leqslant
	\sqrt{\mathbb{E}(\langle \Delta_n^2\rangle-\langle \Delta_n\rangle^2)}\,,
	\]
	and after using the convergences mentioned above one finds that
	$\nu\big(|\Delta_n-\langle \Delta_n\rangle|\big)$ converges to 0 as $n\to \infty$. 
	Since $(\beta,h)\in {\mathcal A}^c$, the proof follows by Remark \ref{hnlmm}.
\end{proof}
	\begin{remark}
		In \cite{Auffinger-other}, \cite{PANCHENKO2010189} and \cite{panchenko2013sherrington}, 
		by using different techniques, the authors proved the following general result
		\[
		\textstyle
		\nu\big(\big|{H_n\over|V_n|} - \nu({H_n\over|V_n|}) \big| \big)  \nto 0\,,
		\]
		where $H_n$ is the Hamiltonian \eqref{hamiltonian} of the RFIM.
		Under an ergodic or mixing hypothesis a straightforward computation shows that this result implies our Lemma \ref{lemma-main}.
	\end{remark}
\begin{remark}
	Taking $h_x=\pm 1$ and $\zeta_x\sim N(0,1)$ for all $x$, in Lemma \ref{lemma-main}, 
	we recovered the proof of Lemma 2.9 in 
	\cite{chatterjee2015absence} 
	using classic inequalities and the essential inequality \eqref{2-ident}
	instead 
	of using the Hermite polynomials, as was done in Lemma 2.8 of \cite{chatterjee2015absence}.
\end{remark}
%
%
%
%
Take any integer $m\geqslant 2$ and let $\sigma^1,\ldots, \sigma^m, \sigma^{m+1}$ 
denote $m+1$ spin configurations drawn independently from the Gibbs measure.  
Let $R_{l,s}$ the overlap between $\sigma^l$ and $\sigma^s$ defined in \eqref{overlap}, 
with $l,s=1,\ldots,m+1$. 
Let $f:\mathbb{R}^{m(m-1)/2}\ra[-1,1]$ be a bounded measurable function
of these overlaps that not change with $n$.
%
\begin{lemma}[Ghirlanda-Guerra identities]\label{ggid}
	Consider the RFIM defined by the Gibbs measure in \eqref{gibbs-measure2}. Then,
	the identity \eqref{G-G}
	is satisfied at almost all $(\beta,h)$. That is, if $f$ is as above, 
	\begin{align*} 
	\nu(f R_{1,m+1})
	- 
	\frac{1}{m}  \,
	\nu(f)
	\nu(R_{1,2})
	- 
	\frac{1}{m}\,
	\sum_{s=2}^m \nu(f R_{1,s})
	\nto 0\,,
	\end{align*}
	for each $(\beta,h)$ in ${\mathcal A}^c$.
\end{lemma}
\begin{proof}
	Let $\langle\,\cdot\,\rangle_{g_x=u}$ be the Gibbs expectation defined by setting 
	$g_x$ in $\langle\,\cdot\,\rangle$ to be $u$ and
	$F_x(u)\coloneqq \langle\,\sigma_x^1 \, f\,\rangle_{g_x=u}$.
	Using \eqref{en-int},
	a  straightforward calculus show that 
	\begin{align}\label{in-princ}
	{\partial^j F_x(u)\over \partial u^j}
	=
	h^j\,
	\biggl\langle\sigma_x^{1} \boldsymbol{\cdot} \Big(\sum_{s=1}^m\sigma_{x}^s-m\sigma_x^{m+1}\Big)^j f \biggr\rangle_{g_x=u},
	\quad j=1,2,\ldots\,.
	\end{align}
	A generalized Gaussian  integration by 
	parts (see Proposition \ref{prop-cheng} in Appendix), with 
	$f_x(u)\coloneqq \mathbb{E}F_x(u)$, gives
	\begin{align*}
	\mathbb{E}g_x f_x - h_x^2\,\mathbb{E}\dfrac{\text{d} f_x}{\text{d} g_x}
	=
	\boldsymbol{\gamma}^2_{g_x}(f_x)\,,
	\end{align*}
	where 
	$
	\boldsymbol{\gamma}^2_{g_x}(f_x)
	=
	\mathbb{E}\big(
	g_x\int_0^{g_x}(g_x-u) \,{\text{d}^2 f_x(u)\over\text{d} u^2} \,\text{d}u\big)
	-
	h_x^2 \,
	\mathbb{E}\big(
	\int_0^{g_x}(g_x-u) \,{\text{d}^3 f_x(u)\over\text{d} u^3} \,\text{d}u\big).
	$
	Dividing the above equality by $|V_n|$ and summing over all $x\in V_n$, and using \eqref{in-princ}
	with $j=1$, we have	
	\begin{align}\label{prr-n}
	\nu\big(\Delta_n(\sigma^1) f\big)
	- 
	h\,\nu\biggl(\Big(\sum_{s=1}^{m}R_{1,s}-mR_{1,m+1}\Big)f \biggr)	 
	=
	{1\over |V_n|}
	\sum_{x\in V_n} 
	\boldsymbol{\gamma}^2_{g_x}(f_x)	\, .
	\end{align}
	Since $|{\partial^j F_x(u)\over \partial u^j}|\leqslant(2mh)^j\|f\|_\infty$, 
	it follows that
	$
	\big|\int_0^{g_x}(g_x-u) \, {\text{d}^3 f_x(u)\over\text{d} u^3} \,\text{d}u\big|
	\leqslant 
	4(mh)^3 \|f\|_\infty h_x^2 \zeta_x^2
	$
	and, by Item \eqref{second-id} of Proposition \ref{prop-cheng}, that
	\begin{align*}
	\textstyle
	\Big| g_x\int_{0}^{g_x}(g_x-u)\, {\text{d}^2 f_x(u)\over\text{d} u^2} \,\text{d}u \Big|
	\leqslant
	2(mh)^2\|f\|_\infty |h_x|^3 |\zeta_x|^3 \, .
	\end{align*}
	Then,
	\begin{align*}
	& \textstyle
	\limsup_{n\to\infty}
	\sup_{f}
	\Big|
	{1\over |V_n|}
	\sum_{x\in V_n}
	\boldsymbol{\gamma}^2_{g_x}(F_x)
	\Big|
	\\[0,1cm]
	& \textstyle
	\leqslant 
	\lim_{n\to\infty}
	{(1+mh)(2mh)^2\|f\|_\infty\over |V_n|}
	\sum_{x\in V_n} 
	|h_x|^3
	\mathbb{E}|\zeta_x|^3=0\, .
	\end{align*}
	Here, the last equality follows from \eqref{condition-main} and of the assumption that
	the $\zeta_x$'s are identically distributed and satisfy $\mathbb{E}|\zeta_x|^3<\infty$.
	Therefore, in \eqref{prr-n}, follows that
	\begin{align}\label{approx}
	\limsup_{n\to\infty}
	\sup_{f}
	\left|
	\nu\big(\Delta_n(\sigma^1) f\big)
	-
	h\,\nu\biggl(\Big(\sum_{s=1}^{m}R_{1,s}-mR_{1,m+1}\Big)f \biggr)	
	\right|=0\,.
	\end{align}
	Since 
	$
	\nu\big(|\Delta_n - \nu(\Delta_n) | \big)  \nto 0
	$
	for any $(\beta,h)\in \mathcal{A}^c$
	(see Lemma \ref{lemma-main}), it is well-known 
	(see e.g. \cite{talagrand2003spin}, Section 2.12) that 
	\eqref{approx} is sufficient to guarantee  the validity of 
	the Ghirlanda-Guerra identities \eqref{G-G}.
	The proof of lemma is complete.
\end{proof}
%
%
%
\vspace*{0,3cm}
\noindent
{\it Proof of Theorem \ref{rsbthm}.}
The proof follows the same path as in \cite{chatterjee2015absence} and we present it  
for the sake of completeness.	
Let $q_{\beta, h,n}\coloneqq \nu(R_{1,2})$.
Taking $f=1$ and $m=1$ in \eqref{prr-n} we obtain
\begin{align}\label{ehn}
\textstyle
\nu\big(\Delta_n(\sigma^1)\big)
= 
h (1-q_{\beta, h,n})
+
\frac{1}{|V_n|}
\sum_x 
\boldsymbol{\gamma}^2_{g_x}(f_x),
\quad
\frac{1}{|V_n|}
\sum_x 
\boldsymbol{\gamma}^2_{g_x}(f_x)\nto 0\, . 
\end{align}
On the other hand,
choosing $m=2$ and $f = R_{1,2}$ in Lemma \ref{ggid} gives 
\begin{equation}\label{gg1}
\nu(R_{1,2}R_{1,3})
- 
\frac{1}{2}  
q_{\beta, h,n}^2
- 
\frac{1}{2}
\nu(R_{1,2}^2)
\nto 0\, . 
\end{equation}
Choosing $m=3$ and $f= R_{2,3}$ gives
\begin{align}\label{gg2}
\nu(R_{2,3} R_{1,4}) 
- 
\frac{1}{3}  
q_{\beta, h,n}^2
- 
\frac{1}{3}
\sum_{s=2}^{3}
\nu(R_{2,3}R_{1,s})
\nto 0\, . 
\end{align}
By symmetry between replicas,
$
\nu(R_{2,3}R_{1,2})
= 
\nu(R_{2,3}R_{1,3})
= 
\nu(R_{1,2}R_{1,3}) \, ,
$
then,
we can multiply \eqref{gg1} by ${2/ 3}$ and add to \eqref{gg2} to get
\begin{align}\label{limit-final}
{2\over 3}
\big(
\nu(R_{1,2}^2)
-
q_{\beta, h,n}^2
\big)
-
\mathbb{E}\big(\smallavg{R_{1,2}^2}
- 
\smallavg{R_{2,3}R_{1,4}} 
\big)
\nto 0\, .
\end{align}
Seeing that the sequence $(\sigma^{l})$ is an independent sequence under
Gibbs' measure,
$
\smallavg{R_{2,3}R_{1,4}} 
= 
{1\over |V_n|^2}\sum_{x,y\in V_n}h_x^2 h_y^2 
\, \langle\sigma^2_x \sigma^3_x \sigma^1_y \sigma^4_y\rangle
=
\smallavg{R_{1,2}}^2\,.
$
Combining this with \eqref{limit-final} and after using
Lemma \ref{mainlmm-il} we have
$
\nu(R_{1,2}^2)-q_{\beta, h,n}^2\nto 0.
$
Therefore,
\begin{align}\label{lim-1}
\nu\big((R_{1,2}-q_{\beta, h, n})^2\big)
=
\nu\big( (R_{1,2}-\nu(R_{1,2}))^2 \big) 
=
\nu(R_{1,2}^2)
-
q_{\beta, h,n}^2
\nto 0\, .  
\end{align}
%
Note that by \eqref{ehn} and by Remark \ref{hnlmm}, 
$
q_{\beta, h}\coloneqq\lim_{n\ra\infty} q_{\beta, h,n}=1-\frac{1}{h}\fpar{p}{h}(\beta,h)
$
exists.
Therefore, taking $n\ra \infty$ in the inequality
\[
\nu\big((R_{1,2}-q_{\beta, h})^2\big)
\leqslant 
2\nu\big((R_{1,2}-q_{\beta, h,n})^2\big)
+
2(q_{\beta, h,n}-q_{\beta, h})^2
\]
and using \eqref{lim-1}, the proof of Theorem \ref{rsbthm} follows. \qed
\section*{Appendix}
The proof of the next result appears in Chen (2019) \cite{chen2019}, Proposition 6.1. 
%
\begin{proposition}\label{prop-cheng}
	Let $Y$ be a real-valued random variable with zero-mean and finite-variance $\sigma^2$,
	with $\sigma>0$. For any 
	function $f:\mathbb{R}\to\mathbb{R}$
	with a bounded continuous third-order derivative, we have 
	\begin{align}\label{first-id}
	\mathbb{E}Yf(Y)
	&=
	\sigma^2\, \mathbb{E}f'(Y)
	+
	\boldsymbol{\gamma}^2_{Y}(f)\,,
	\end{align}
	where
	\begin{align*}
	\boldsymbol{\gamma}^2_{Y}(f)
	\coloneqq
	\textstyle
	\mathbb{E}\Big(
	Y\int_0^{Y}(Y-u) f''(u) \,{\rm d}u\Big)
	-
	\sigma^2\,
	\mathbb{E}\Big(
	\int_0^{Y}(Y-u) f'''(u) \,{\rm d}u\Big)\,,
	\end{align*}
	with
	\begin{align}\label{second-id}
	\textstyle
	\Big| Y\int_{0}^{Y}(Y-u) f''(u) \,{\rm d}u \Big|
	\leqslant
	|Y|
	\int_0^{|Y|}\min\left\{2\,\|f'\|_\infty, 
	\|f''\|_\infty \,u\right\}
	\,{\rm d}u\,.
	\end{align}
\end{proposition}
%
%
The next result is new and can be seen as a generalization of Proposition \ref{prop-cheng} for the bivariate case. 
In order to lighten the notation we will write
$\partial_{i,j} f$ 
to denote
the partial derivative of order 
$i$ and $j$ for the first and second component respectively. 
\begin{proposition}[A generalized Gaussian integration by parts]\label{prop-cv}
	Let $X$ and $Y$ be two independent real-valued random variables with zero-mean and finite-variances $\sigma_X^2$ and $\sigma_Y^2$ respectively,
	with $\sigma_X$, $\sigma_Y$ both positive. For any 
	function $f:\mathbb{R}^2 \to\mathbb{R}$
	with a bounded continuous fifth-order derivative, we have
	\begin{align}\label{first-id-a}
	\mathbb{E}XYf(X,Y)
	=
	\sigma_X^2 \sigma_Y^2\, \mathbb{E}\partial_{1,1} f(X,Y)
	+
	\boldsymbol{\gamma}^2_{X,Y}(f)
	\,,
	\end{align}
	where 
	\begin{align*}
	\textstyle
	\boldsymbol{\gamma}^2_{X,Y}(f)
	&\coloneqq
	\textstyle
	\mathbb{E}\Big(
	XY\int_0^{X} \int_0^{Y} (X-u)\, \partial_{2,1}f(u,v) \,{\rm d}u {\rm d}v \Big)
	\\[0,1cm]
	&\quad -
	\textstyle
	\sigma_X^2 \sigma_Y^2  \,
	\mathbb{E}\Big(
	\int_0^{X} \int_0^{Y} (X-u)\, \partial_{3,2}f(u,v) \,{\rm d}u {\rm d}v
	\Big)
	\\[0,1cm]
	&\quad - 
	\textstyle
	\sigma_X^2 \sigma_Y^2\,
	\mathbb{E}\Big(
	\int_{0}^{X}(X-u)\partial_{3,1}f(u,0) \,\text{d}u
	+
	\int_{0}^{Y}(Y-v)\partial_{1,3}f(0,v) \,\text{d}v
	\Big)\,.
	\end{align*}
	Furthermore,
	\begin{align}\label{second-id-1}
	& \textstyle
	\Big| 
	XY\int_0^{X} \int_0^{Y} (X- u)\, \partial_{2,1}f(u,v) \,{\rm d}u {\rm d}v
	\Big|
	\\[0,1cm]
	& \hspace*{2cm}\leqslant
	\textstyle
	|X||Y|
	\int_0^{|X|} \int_0^{|Y|} \min\big\{2\,\|\partial_{1,1} f\|_\infty\, , 
	\|\partial_{2,1} f\|_\infty \,u\big\}
	\,{\rm d}u {\rm d}v\,,
	\nonumber
	\\[0,1cm]
	& \textstyle
	\Big| 
	\int_{0}^{X}(X-u)\,\partial_{3,1}f(u,0) \,\text{d}u
	\Big|
	\leqslant
	\int_0^{|X|} \min\big\{2\,\|\partial_{2,1} f\|_\infty\, , 
	\|\partial_{3,1} f\|_\infty \,u\big\}
	\,{\rm d}u \,,  \label{second-id-2}
	\\[0,1cm]
	& \textstyle
	\Big| 
	\int_{0}^{Y}(Y-v)\,\partial_{1,3}f(0,v) \,\text{d}v
	\Big|
	\leqslant
	\int_0^{|Y|} \min\big\{2\,\|\partial_{1,2} f\|_\infty\, , 
	\|\partial_{1,3} f\|_\infty \,v\big\}
	\,{\rm d}v \,. \label{second-id-3}
	\end{align}
\end{proposition}
\begin{proof}
	Taylor's Theorem for multivariate functions gives,
	\begin{align*}
	& XYf(X,Y)
	\\[0,1cm]
	&= \textstyle
	XYf(0,0)
	+
	X^2Y \partial_{1,0}f(0,0)
	+
	X Y^2 \partial_{0,1}f(0,0)
	\\[0,1cm]
	&\quad + \textstyle
	{X^3Y\over 2}\, \partial_{2,0}f(0,0)
	+
	{XY^3\over 2}\, \partial_{0,2}f(0,0)
	+
	X^2Y^2 \partial_{1,1}f(0,0)
	\\[0,1cm]
	&\quad + \textstyle
	XY \int_{0}^{X} {(X-u)^2\over 2}\, \partial_{3,0}f(u,0) \,\text{d}u
	+
	XY \int_{0}^{Y} {(Y-v)^2\over 2}\, \partial_{0,3}f(0,v) \,\text{d}v
	\\[0,1cm]
	&\quad + \textstyle
	X^2 Y \! \int_{0}^{Y} (Y-v)\, \partial_{1,2}f(0,v) \,\text{d}v
	+
	XY \!\int_{0}^{X}\!\!\int_{0}^{Y} (X-u)\, \partial_{2,1}f(u,v) \,\text{d}u \text{d}v
	\\[0,1cm]
	&\quad + \textstyle
	\sigma_X^2 \sigma_Y^2
	\Big(
	\partial_{1,1}f(X,Y) - \partial_{1,1}f(0,0) - X \partial_{2,1}f(0,Y) -Y\partial_{1,2}f(0,0)
	\\[0,1cm]
	&
	\quad  \textstyle
	- \int_{0}^{X}(X-u)\,\partial_{3,1}f(u,0) \,\text{d}u
	- \int_{0}^{Y}(Y-v)\,\partial_{1,3}f(0,v) \,\text{d}v
	+
	R(X,Y)
	\Big)
	\,,
	\end{align*}
	where
	\begin{align*}
	R(X,Y)
	&\coloneqq \textstyle
	-\partial_{1,1}f(X,Y) + \partial_{1,1}f(0,0) 
	+ X \partial_{2,1}f(0,Y) +Y\partial_{1,2}f(0,0)
	\\[0,1cm]
	&\quad + \textstyle
	\int_{0}^{X}(X-u)\,\partial_{3,1}f(u,0) \,\text{d}u
	+ \int_{0}^{Y}(Y-v)\,\partial_{1,3}f(0,v) \,\text{d}v.
	\end{align*}
	A simple observation shows that
	\begin{align}\label{id-w}
	R(X,Y)
	&= \textstyle
	\int_{0}^{X}(X-u)\big(\partial_{3,1}f(u,Y)-\partial_{3,1}f(u,0)\big) \,{\rm d}u
	\nonumber
	\\[0,1cm]
	&= \textstyle
	-\int_0^{X} \int_0^{Y} (X-u)\, \partial_{3,2}f(u,v) \,{\rm d}u {\rm d}v\,.
	\end{align}
	Since $X$, $Y$ are independent random variables, $\mathbb{E}X=\mathbb{E}Y=0$ and $\mathbb{E}X^2=\sigma_X^2$, $\mathbb{E}Y^2=\sigma_Y^2$, one finds that
	\begin{align*}
	&\mathbb{E}XYf(X,Y)
	\\[0,1cm]
	&= \textstyle
	\sigma_X^2 \sigma_Y^2\, \mathbb{E}\partial_{1,1} f(X,Y)
	+
	\sigma_X^2 \sigma_Y^2\, \mathbb{E}R(X,Y)
	\\[0,1cm]
	&\quad \textstyle
	+
	\mathbb{E}\Big(
	XY\int_0^{X} \int_0^{Y} (X-u)\, \partial_{2,1}f(u,v) \,{\rm d}u {\rm d}v \Big)
	\\[0,1cm]
	& \quad -  \textstyle
	\sigma_X^2 \sigma_Y^2\,
	\mathbb{E}\Big(
	\int_{0}^{X}(X-u)\,\partial_{3,1}f(u,0) \,\text{d}u
	+
	\int_{0}^{Y}(Y-v)\,\partial_{1,3}f(0,v) \,\text{d}v
	\Big)
	\,.
	\end{align*}
	Then, using \eqref{id-w} the proof of \eqref{first-id-a} follows.
	On the other hand, the Items \eqref{second-id-1}, \eqref{second-id-2} 
	and \eqref{second-id-3} follow by combining each of the following identities
	\begin{align*}
	\textstyle
	\int_0^{X} \int_0^{Y} (X-u)\, \partial_{2,1}f(u,v) \,{\rm d}u {\rm d}v
	&= \textstyle
	\int_0^{X} \int_0^{Y} \big(\partial_{1,1}f(u,v)-\partial_{1,1}f(0,v)\big)  {\rm d}u {\rm d}v\,,
	\\[0,1cm] \textstyle
	\int_{0}^{X}(X-u)\,\partial_{3,1}f(u,0) \,\text{d}u
	&= \textstyle
	\int_0^{X} \big(\partial_{2,1}f(u,0)-\partial_{2,1}f(0,0)\big)  {\rm d}u\,,
	\\[0,1cm] \textstyle
	\int_{0}^{Y}(Y-v)\,\partial_{1,3}f(0,v) \,\text{d}v
	&= \textstyle
	\int_0^{Y} \big(\partial_{1,2}f(0,v)-\partial_{1,2}f(0,0)\big)  {\rm d}v\,,
	\end{align*}
	with the Mean-Value Theorem.

\end{proof}

\section*{Acknowledgements}
It is a pleasure to thank
R. Bissacot,
S. Chatterjee, 
L. Cioletti,
and
L. R. Fontes
for fruitful discussions, questions, references and
helpful suggestions personally and by email, on earlier versions of this manuscript. 
This study was financed in part by the Coordena\c{c}\~{a}o de Aperfei\c{c}oamento de Pessoal de N\'{i}vel Superior - Brasil (CAPES) - Finance Code 001.
Jamer Roldan was supported by CNPq.

\end{document}